\documentclass[11pt,a4paper]{article}
\usepackage[utf8]{inputenc}
\usepackage[T1]{fontenc}

\usepackage[dvipsnames,usenames]{color}
\usepackage[colorlinks=true,urlcolor=Blue,citecolor=Green,linkcolor=BrickRed]{hyperref}
\usepackage[usenames,dvipsnames]{xcolor}
\usepackage{amsthm,amsmath,amssymb, mathabx}
\usepackage{fullpage}
\usepackage[ruled,noend,linesnumbered]{algorithm2e}     
\usepackage{enumerate,comment}
\usepackage{url}
\usepackage[capitalise]{cleveref}
\usepackage{todonotes}
\usepackage{standalone}
\usepackage[noadjust]{cite}
\usepackage{mathtools}
\usepackage{enumerate}
\usepackage{pgfplots}
\usepackage{subcaption}
\usepackage{pdfpages}
\usepackage{xcolor,soul}
\definecolor{light-gray}{gray}{0.85}
\newcommand{\hlgray}[1]{{\sethlcolor{light-gray}\hl{#1}}}

\newtheorem{theorem}{Theorem}[section]

\newtheorem{corollary}[theorem]{Corollary}

\theoremstyle{definition}

\newtheorem{claim}{Claim}

\def\Gf{G\setminus \{f\}}

\def\u{u} 
\def\I{I}

\def\Psf{P'}
\def\Hsf{H'}
\def\Htf{\hat{H}}
\def\Q{Q}
\def\P{P}
\def\Pa{P_1}
\def\Pb{P_2}

\def\Pz{P}
\def\PP{\hat P}

\newcommand{\first}[3]{\mathsf{first_{#1}}(#2,#3)}
\newcommand{\last}[3]{\mathsf{last_{#1}}(#2,#3)}

\title{\~{O}ptimal Fault-Tolerant Reachability Labeling in Planar Graphs}

\begin{document}

\author{Shiri Chechik\thanks{Tel Aviv University, Israel, \href{mailto:schechik@tauex.tau.ac.il}{schechik@tauex.tau.ac.il} } 
\and Shay Mozes\thanks{Reichman University, Israel,  \href{mailto:smozes@idc.ac.il}{smozes@idc.ac.il}} 
		\and Oren Weimann\thanks{University of Haifa, Israel, \href{mailto:oren@cs.haifa.ac.il}{oren@cs.haifa.ac.il}}
} 
		
\date{}
		
\maketitle
	
	\thispagestyle{empty}

\begin{abstract}
We show how to assign labels of size $\tilde O(1)$ to the vertices of a directed planar graph $G$, such that from the labels of any three vertices $s,t,f$ we can deduce in $\tilde O(1)$ time whether $t$ is reachable from $s$ in the graph $\Gf$. Previously it was only known how to achieve $\tilde O(1)$ queries using a centralized $\tilde O(n)$ size oracle [SODA'21].  
\end{abstract}

\newpage
\clearpage
\setcounter{page}{1}

\section{Introduction}
Reachability is undeniably one of the most fundamental properties of graphs. Determining the reachability between all pairs of vertices in a directed graph can be done in $\tilde{O}(\min\{n^{\omega}, mn\})$ time (where $\omega<2.373$ is the exponent for matrix multiplication \cite{alman2021refined}).

\paragraph{Reachability oracles.}
A {\em reachability oracle} is a compact data structure that allows to efficiently answer reachability queries between any pair of vertices in the graph. Henzinger et al. \cite{Henzinger2017} provided conditional lower bounds for combinatorial constructions of reachability oracles, showing that no non-trivial combinatorial reachability oracle constructions exist for general directed graphs.\footnote{The term combinatorial is often referred to algorithms that do not utilize fast matrix multiplications.}
Specifically, they proved that it is impossible to design a reachability oracle that simultaneously achieves $O(n^{3-\epsilon})$ preprocessing time and $O(n^{2-\epsilon})$ query time, for any $\epsilon > 0$.

Since non-trivial reachability oracles are not attainable for general graphs, efforts have been directed towards developing improved reachability oracles for specific graph families. Notably, graphs possessing separators of size $s(n)$ admit a straightforward reachability oracle of size $\tilde O(n\cdot s(n))$ and query time $\tilde O(s(n))$.  
Consequently, planar graphs (and more extensive graph classes such as H-minor free graphs) admit oracles of size $\tilde O(n^{1.5})$ and query time $\tilde O(\sqrt{n})$.
In a groundbreaking result, Thorup \cite{Thorup04} introduced a near-optimal reachability oracle for directed planar graphs with $\tilde O(n)$ space and $\tilde O(1)$ query time. Subsequently, Holm et al. \cite{HolmRT15} further improved this construction to a truly optimal oracle with $O(n)$ space and $O(1)$ query time.

\paragraph{Fault-tolerant reachability oracles.}
In real-world networks, the vulnerability to changes and failures is a significant concern, and effectively handling failures has become crucial in modern computing. This motivation has sparked extensive research on data structures that can operate robustly in the presence of failures.
Specifically, a {\em fault-tolerant} reachability oracle is designed to handle queries of the form $s,t,f$ and determines whether vertex $t$ is reachable from vertex $s$ in the graph $\Gf$ (the graph $G$ with vertex $f$ and all its incident edges removed). 

Fault-tolerant reachability oracles have been studied extensively in general graphs, see e.g. \cite{brand2019sensitivity,GeorgiadisIP17,choudhary2016DualFaultTolerant,BaswanaCR18,king2002fully,GeorgiadisGIPU17}. In planar graphs, one can leverage the more powerful fault-tolerant {\em distance} oracles. Baswana et al. \cite{Baswana} introduced a single-source fault-tolerant distance oracle with near-optimal $\tilde O(n)$ space and $\tilde O(1)$ query time.
They further extended their construction to handle the all-pairs variant of the problem at the cost of an increased $\tilde O(n^{1.5})$ space and $\tilde O(\sqrt{n})$ query time.
Subsequently, Charalampopoulos et al.~\cite{faultyOracle} presented an improved fault-tolerant distance oracle for the all-pairs version in planar graphs. Their fault-tolerant oracle accommodates multiple failures; however, it comes with a polynomial trade-off between the size of the oracle and the query time that may be less favorable compared to previous constructions.
 It is important to highlight that, in the context of the all-pairs version, the fault-tolerant oracles mentioned above have considerably worse bounds when compared to the best known distance oracles without faults for planar graphs. 
Finally, in a groundbreaking development, Italiano et al. \cite{Reachability} in SODA 2021 introduced a nearly optimal fault-tolerant reachability oracle for directed planar graphs. Their innovative approach finally achieved near-optimal $\tilde O(n)$ size and $\tilde O(1)$  query time.

\paragraph{Fault-tolerant reachability labeling.}
A more powerful concept than oracles is that of {\em labeling schemes}, where each vertex is assigned a compact label, and the objective is to answer queries based solely on the labels of the involved vertices. Labeling schemes have proven to be valuable in distributed settings where inferring properties like distances or reachablity is advantageous using only local information such as the labels of the source and destination vertices. Such scenarios arise in communication networks or disaster-stricken areas, where communication with a centralized entity may be infeasible or impossible.
Labeling schemes provide a natural and structured framework for studying the distribution of graph information. Various problems have been explored within this model, including adjacency  \cite{Kannan,alstrup2015optimal,petersen2015near,alstrup2015adjacency,AlonN17,bonichon2007short}, distances  \cite{GPPR04,Bar-NatanCGMW22,Thorup04,AbrahamCG12,GavoilleKKPP01}, flows and connectivity \cite{KatzKKP04,HsuL09,Korman10}, and Steiner tree \cite{Peleg05}. See~\cite{rotbart2016new} for a survey.

In this paper, we focus on fault-tolerant reachability labeling in directed planar graphs. The objective is to assign a compact label to each vertex, such that given the labels of any three vertices $s,t,f$, one can efficiently determine whether $t$ is reachable from  $s$ in the graph $\Gf$.

The best previously known labeling scheme for fault-tolerant reachability was the one for fault-tolerant distances by \cite{Bar-NatanCGMW22} in which the label size is $\tilde{O}(n^{2/3})$.
Is this the best possible?
Ideally, the goal would be to devise a labeling scheme in which  the sum of the label sizes is roughly equal to the size of the state-of-the-art oracle. However, achieving this goal is not always possible.
For example, in \cite{CharalampopoulosGLMPWW23}, an almost optimal exact distance oracle is given for directed planar graphs (without faults) of size $O(n^{1+o(1)})$ and query time $\tilde{O}(1)$. On the other hand, it was shown in \cite{GPPR04} that exact distance labels (even without faults) for planar graphs necessitate polynomial-sized labels regardless of the query time.
Given this discrepancy between oracles and labeling schemes for distances in planar graphs, a natural question arises: does the same discrepancy exist for fault-tolerant reachability?
In other words, is it possible to design a labeling scheme for directed planar graphs with label size $\tilde{O}(1)$ and query time $\tilde{O}(1)$? Or does a similar gap exist between fault-tolerant reachability oracles and labeling schemes, as in the case of distances? 

\paragraph{Our results.} We answer this question in the affirmative by providing a near optimal labeling scheme for fault-tolerant reachability in directed planar graphs with $\tilde{O}(1)$ label size and $\tilde{O}(1)$ query time.  

To obtain a fault-tolerant reachability labeling scheme, one might hope to generalize existing labeling schemes for {\em undirected} graphs. 
Specifically, for undirected planar graphs, Abraham et al.~\cite{AbrahamCG12} presented labels of size $\tilde{O}(1)$ that for any fixed $\epsilon>0$, from the labels of vertics $s,t,$ and the labels of a set $F$ of failed vertices, can report in $\tilde O(|F|^2)$ time a $(1+\epsilon)$-approximation of the shortest $s$-to-$t$ path in the graph $G\setminus F$.
One would hope to generalize this result to the directed case, even just settling for the seemingly easier task of reachability, and even for a single fault.
This is particularly tempting because previous results for the failure-free case have successfully generalized from undirected to directed planar graphs.
For example, Thorup~\cite{Thorup04} was able to convert his $(1+\epsilon)$-distance oracle for undirected planar graphs to also work in the directed case by employing clever ideas. However, in the presence of failures, the task becomes significantly more challenging.
In a nutshell, one of the challenges in adapting undirected results to the directed case in the context of fault-tolerant reachability labeling in planar graphs is the following.
In many planar graph papers, including \cite{Thorup04} and \cite{AbrahamCG12}, the algorithms store distances from a vertex to other vertices in relevant separators. In Thorup's reachability oracle~\cite{Thorup04}, the algorithm stores, for each vertex $s$ and each relevant path separator, the first vertex on the path that is reachable from $s$ and the last vertex on the path that can reach $s$.
To determine if vertex $s$ can reach vertex $t$ by a path that intersects the path separator, one can examine the information stored in the labels of $s$ and $t$. By utilizing this stored information, it becomes possible to check the reachability between $s$ and $t$.
One of the main challenges we face with this approach is the occurrence of failures anywhere along the relevant path separator. A faulty vertex $f$ on the path separator requires us to store additional information, including the closest vertex to $f$ that appears after $f$ on the path separator and is reachable from the starting vertex $s$. Considering that failures can happen at any vertex on the path separator, this means we would need to store all vertices that are reachable from $s$ and to $s$ on the path separator. This requirement renders the approach impractical.

To address this issue, we need to adopt a different approach and develop new techniques specifically tailored for accommodating failures in the directed case.

It is worth mentioning that in both our algorithm and the algorithm proposed in \cite{Reachability}, the most challenging scenario arises when all three vertices $s,t,$ and $f$ are situated on the same path separator $P$.
In \cite{Reachability}, this situation was addressed by employing a data structure that extends dominator trees and previous data structures designed for handling strong-connectivity in general (non-planar) graphs under failures \cite{GeorgiadisIP17}. We tackle this case without relying on complex dominator trees or similar sophisticated techniques.
Therefore, our algorithm not only achieves the milestone of introducing a near optimal efficient fault-tolerant labeling scheme for planar reachability, but it also boasts a simpler approach compared to the one employed in \cite{Reachability}. We firmly believe that the simplicity of our algorithm  makes it an important milestone in generalization to labels for approximate distances and for multiple failures in directed planar graphs and related graph families.

\section{Preliminaries}

\paragraph{The decompostition tree.} 
In~\cite[Lemma 2.2]{Thorup04}, Thorup proved that we can assume the graph $G$ has an undirected spanning tree $T$ (i.e., $T$ is an unrooted spanning tree in the undirected graph obtained from $G$ by ignoring the directions of edge) such that each path in $T$ is the concatenation of $O(1)$ directed paths in $G$. 

This way, we can describe the process of decomposing $G$ into pieces in the undirected version of $G$. After describing the decomposition, we will replace each undirected path of $T$ defined in the process by its $O(1)$ corresponding directed paths in $G$. We therefore proceed to describe the decomopostion treating $G$ as an undirected graph with a rooted spanning tree $T$.  

A balanced simple cycle separator~\cite{LTsep} (cf. \cite[Lemma 2.3]{Thorup04}) is a simple cycle $C$ in $G$ whose vertices can be covered by a single path of the (unrooted) spanning tree $T$. The removal of the vertices of $C$ and their incident edges separates $G$ into two roughly equal sized subgraphs.
The recursive decomposition tree $\mathcal T$ of $G$ is defined as follows.
Each node of $\mathcal T$ corresponds to a vertex induced subgraph of $G$ (called a {\em piece}). The root piece of $\mathcal T$ is the entire graph $G$. 
The boundary $\partial H$ of a piece $H$ is a set of vertex disjoint paths that lie on the faces of $H$ that are not faces of $G$.
A vertex is a {\em boundary vertex} if it belongs to some path of $\partial H$. 
The boundary of the root piece is empty.
We define the children of a piece $H$ in $\mathcal T$ with boundary $\partial H$ and a simple cycle separator $C$ recursively. Let $Q$ be the maximal subpath of $C$ that is internally disjoint from $\partial H$.  
The vertices of $H$ that are enclosed by $C$ (including the vertices of $C$) belong to one child of $H$. The vertices of $Q$ and the vertices of $H$ not enclosed by $C$ belong to the other child. 
Note that the vertices of $Q$ are the only vertices of $H$ that belong to both children.
The endpoints of $Q$ that belong to $\partial H$ are called the {\em apices} of $H$.
We call the path $Q$ without the apices of $H$ the {\em separator} of $H$. We do not include the apices in the separator to guarantee that it is vertex disjoint from $\partial H$. 
The boundary $\partial H'$ of a child $H'$ of $H$ consists of the separator of $H$ and of the subpaths of $\partial H$ induced by the vertices of $H'$.
The leaves of $\mathcal T$ (called {\em atomic} pieces) correspond to pieces of size $O(1)$. The depth of $\mathcal T$ is $O(\log n)$. For convenience, we consider all $O(1)$ vertices of an atomic (leaf) piece that are not already boundary vertices as the separator of the piece.
It follows that the boundary $\partial H$ of any piece $H$ consists of  $O(\log n)$ vertex disjoint paths (the subpath induced by the vertices of $H$ on the separators of the of the ancestor pieces of $H$). Also, there are $O(\log n)$ apices along any root-to-leaf path in $\mathcal T$.

Having defined the decomposition tree $\mathcal T$ we can go back to treating $G$ as a directed graph. As we explained above, each path we had discussed in the undirected version of $G$ is the union of $O(1)$ directed path in $G$. From now on when we refer to the separator paths of a piece $H$ (resp., paths of $\partial H$), we mean the set of directed paths comprising the undirected separator of $H$ (resp., the set of directed paths comprising the paths of $\partial H$). 

To be able to control the size of the labels in our construction we need to be aware of the number of pieces of $\mathcal T$ to which a vertex belongs. 
The only vertices of a piece $H$ that belong to both its children are the vertices of the separator path of $H$ and the (at most 2) apices of $H$. 
The above definitions imply that every vertex belongs to the separator of at most one piece in $\mathcal T$. Hence, if a vertex is not an apex, it appears in $O(\log n)$ pieces of $\mathcal T$. 
Apices, on the other hand require special attention because they may belong to many (i.e., $\omega(\log n)$) pieces of $\mathcal T$; High degree vertices may be apices in many pieces, and we will need a special mechanism for dealing with such vertices. 
Dealing with apices (like dealing with holes in other works on planar graphs) introduces technical complications that are not pertinent to understanding the main ideas of our work.\footnote{A reader who is not interested in those details can safely skip the parts dealing with apices and just act under the assumption that each vertex appears in 2 atomic pieces (leaves) of $\mathcal T$, and that the following definition of ancestor pieces of a vertex $v$ just degenerates to the set of $O(\log n)$ ancestors of the 2 atomic pieces containing  $v$.}

We associate with every vertex $v\in G$ the (at most 2) rootmost pieces $H$ in $\mathcal T$ in which $v$ is an apex (or the atomic pieces containing $v$ if $v$ is never an apex). We denote these pieces by $H_v$. Note that every piece that contains a vertex $v$ is either an ancestor of a piece in $H_v$ or a descendant of a piece in $H_v$.
For a vertex $v\in G$ we define the {\em ancestor pieces} of $v$  to be the set of (weak) ancestors in $\mathcal T$ of the pieces $H_v$. 
By definition of $H_v$, every vertex, apex or not, has $O(\log n)$ ancestors pieces. 
We similarly define the ancestor separators/paths/apices of a vertex $v\in G$ as the separators/separator-paths/apices of any ancestor piece of $v$. See Figure~\ref{fig:decomposition}.

\begin{figure}[h]
\begin{center}
\includegraphics[scale=0.18]{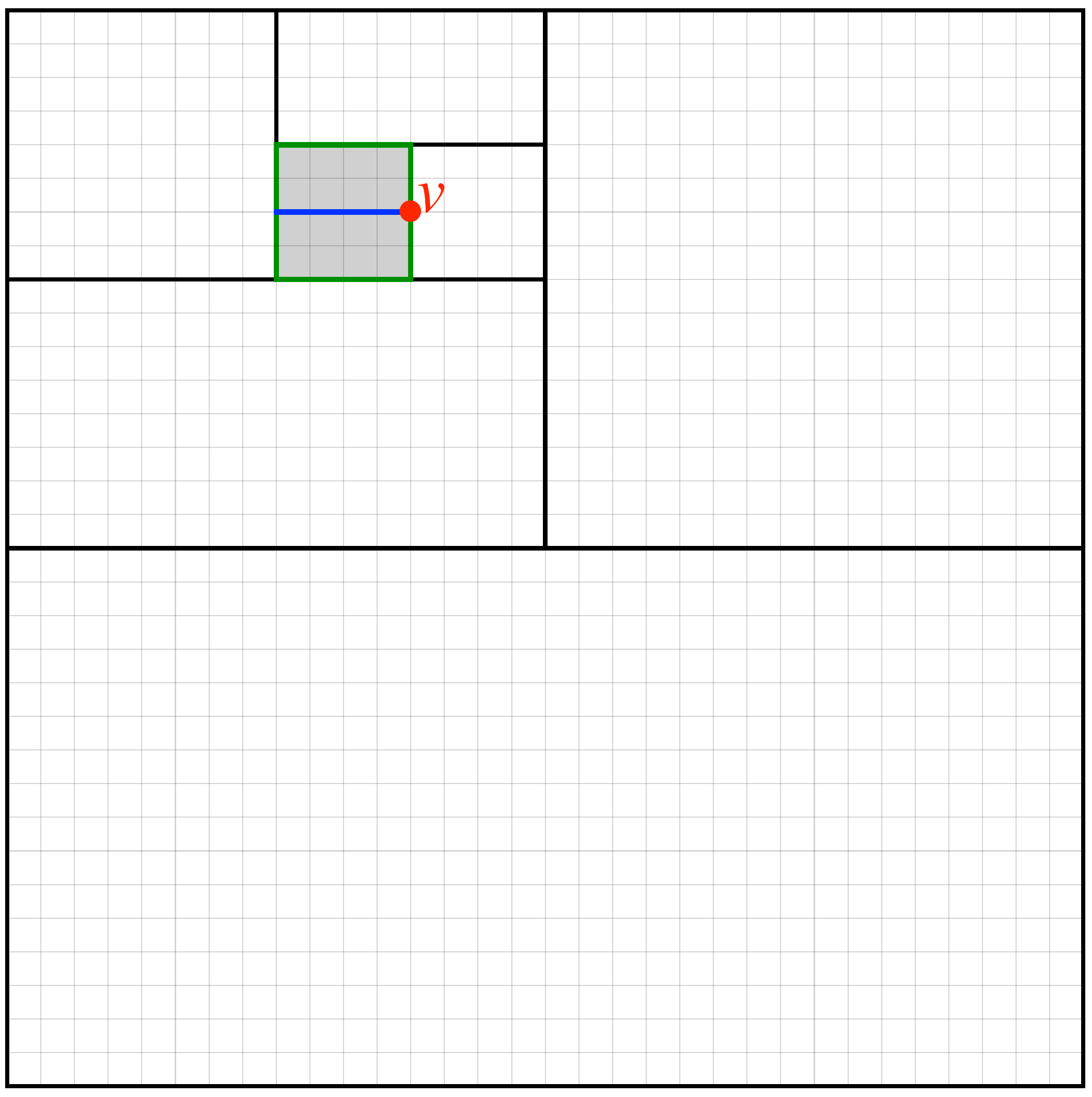}
\caption{Pieces in a recursive decomposition. The separators in this example alternating between horizontal and vertical lines. 
All rectangular pieces that contain the gray piece $H$ are its ancestors.
The boundary $\partial H$ of $H$ is shown in green.
The maximal subpath of the cycle separator of $H$ that is internally disjoint from $\partial H$ is shown in blue.
The vertex $v$ is an apex of $H$ because it is an endpoint of this maximal subpath. The separator of $H$ is the blue path (without its endpoints). 
\label{fig:decomposition}
} 	
\end{center}
\end{figure}

We say that the separator $Q$ of a piece $H$ {\em separates} two vertices $u$ and $v$ (in $H$) if any $u$-to-$v$ path in $H$ must touch the separator $Q$ of $H$ or an apex of $H$. I.e., if at least one of the following holds: (1) $u\in Q$ or $u$ is an apex of $H$, or (2) $v\in Q$ or $v$ is an apex of $H$, or (3) each of $u$ and $v$ is in one distinct child of $H$. 
Note that if $Q$ separates $u$ and $v$ in $H$ then every $u$-to-$v$ path in $G$ either touches $Q$ or touches the boundary of $H$. 

For a subgraph $H$, a path $P$ and a vertex $v$, let $\first{H}{v}{P}$ denote the first vertex of $P$ that is reachable from $v$ in $H$, and let $\last{H}{v}{P}$ denote the last vertex of $P$ that can reach $v$ in $H$. If vertex $u$ appears before vertex $v$ on a path $P$ then we denote this by $u<_P v$ (or simply $u<v$ if $P$ is clear from the context). 
Throughout the paper, we gradually describe the information stored in the labels along with the explanations of why this particular information is stored (and why it is polylogarithmic). To assist the reader, we highlight in gray the parts that describe the information stored. For starters, we let 
\hlgray{
every vertex $v\in G$ store in its label, for every ancestor path $P$ of $v$, the identity of $P$ and, if $v \in P$, the location of $v$ in $P$} (so that given two vertices $u,v$ of $P$, we can tell if $u<v$). We denote $P[u,v]$ the subpath of $P$ between vertices $u$ and $v$.

\paragraph{Thorup's non-faulty labeling.}
Using the above definitions and notations, it is now very simple to describe Thorup's non-faulty reachability labeling~\cite{Thorup04}. 
Consider any vertex $v$. Let $H$ be the rootmost piece in $\mathcal T$ in which $v$ belongs to the separator. 
The crucial observation is that $v$ is separated from every other vertex in $G$ either by the separator of $H$ or by the separator of some ancestor piece of $H$.
Hence, every vertex $v\in G$ stores in its label $\first{G}{v}{P}$ and $\last{G}{v}{P}$ for every path $P$ of the separator of every ancestor of the rootmost piece in which $v$ belongs to the separator. 
Then, given a query pair $u,v$, there exists a $u$-to-$v$ path in $G$ if and only if $\first{G}{u}{P} < \last{G}{v}{P}$ for one of the $O(1)$ paths $P$ of the separator of an ancestor piece of the rootmost piece whose separator separates $u$ and $v$. Both $u$ and $v$ store the relevant information for these paths in their labels. We note that in Thorup's scheme we do not need to worry about apices since each vertex $v$ only stores information in pieces above the first time $v$ appears on a separator. 

\section{The Labeling}

In this section, we explain our labeling scheme. I.e., what to store in the labels so that given the labels of any three vertices $s,f,t$ we can infer whether $t$ is reachable from $s$ in $\Gf$. We call the $s$-to-$t$ path $R$ in $\Gf$ the {\em replacement path}.

Let $\Htf$ be the rootmost piece in $\mathcal T$ whose separator $Q$ separates $t$ and $f$. 
Let $H$ be child piece of $\Htf$ that contains $t$ (if both children of $\Htf$ contain $t$ then, if one of the children does not contain $f$ we choose $H$ to be that child). 
Note that by choice of $H$, $f \notin H \setminus \partial H$. 
We assume without loss of generality that $s \in \Htf$. We handle the other case by storing a symmetric label to the one described here in the graph $G$ with all edges reversed (the reverse graph of $G$ has exactly the same decomposition tree as $G$, but there the roles of $s$ and $t$ are swapped, so the assumption does hold).  
Observe that by definition of $H$ and of separation, $f \in \partial H$ iff $f \in \Q$. 
In what follows, we separately handle the cases where $f \notin \Q$ and $f \in Q$. The main challenge is in the latter case.   
 
\subsection{When $f \notin \Q$ (and so, $f \notin \partial H$)} \label{sec:fnotinQ}

Consider first the case when the replacement path $R$ does not touch $\partial H$. i.e., $s,t$ and $R$ are all contained in $H\setminus \partial H$. 
We handle this case by having \hlgray{each vertex $v \in H \setminus \partial H$ (and in particular $s$ and $t$) store the standard (non-faulty) labeling of Thorup for $H\setminus \partial H$.} 
Since each vertex $v \in G$ is non-boundary in $O(\log n)$ pieces of $\mathcal T$, this contributes $\tilde O(1)$ to the size of the label at each vertex. 

We next consider the case that $R$ touches $\partial H$. In this case, $R$ must have a suffix contained in $H$, and this suffix is unaffected by the fault $f$. More precisely, $R$ exists iff $\first{\Gf}{s}{\P} < \last{H}{t}{\P}$  for one of the paths $\P$ forming $\partial H$.
It is easy to find $\last{H}{t}{\P}$; 
\hlgray{For every vertex $v$ in $H$, if $H$ is an ancestor piece of $v$, then $v$ stores in its label $\last{H}{v}{\P}$ for each of the $O(\log n)$ paths $\P$ of $\partial H$.}
Since each vertex has only $O(\log n)$ ancestor pieces, this contributes $\tilde O(1)$ to the size of the label. 
Notice that by the rootmost choice of $H$, $H$ is an ancestor piece of $t$, so $t$ indeed stores $\last{H}{t}{\P}$.
It thus remains only to describe how to find $\first{\Gf}{s}{\P}$ from the labels of $s$ and $f$. 

\hlgray{For every vertex $s \in G$, for every ancestor apex $f$ of $s$ and for every ancestor path $\P$ of $s$, $s$ stores in its label $\first{\Gf}{s}{\P}$. Similarly, for every vertex $f \in G$, for every ancestor apex $s$ of $f$ and for every ancestor path $\P$ of $f$, $f$ stores in its label $\first{\Gf}{s}{\P}$.}

If either $s$ or $f$ stores $\first{\Gf}{s}{\P}$, we are done. 
Otherwise, consider the set of leafmost pieces in $\mathcal T$ that contain both $s$ and $f$. Let $\Hsf'$ be such a piece. It must be that $\Hsf'$ is an ancestor piece of both $s$ and $f$ or else one of $s$ and $f$ is an ancestor apex of the other and stores $\first{\Gf}{s}{\P}$. It follows that if neither $s$ nor $f$ store $\first{\Gf}{s}{\P}$, then there are only $O(1)$ leafmost pieces that contain both $s$ and $f$. 
To avoid unnecessary clutter we shall assume there is a unique piece $\Hsf'$. In reality we would have to apply the same argument for all $O(1)$ such pieces. 
Since $\Hsf'$ is an ancestor piece of both $s$ and $f$, we can find the piece $\Hsf'$ by traversing the list of ancestors of $s$ (stored in $s$) and of $f$ (stored in $f$) until finding the lowest common ancestor. 
Let $\Hsf$ be the child piece of $\Hsf'$ that contains only $s$ (if $\Hsf'$ is an atomic piece then define $\Hsf=\Hsf'$).
Recall that both $s$ and $f$ are in $\Htf$, so $\Htf$ is a (possibly weak) ancestor of $\Hsf'$. 

\begin{figure}[htb]
  \begin{center}
 \includegraphics[scale=0.18]{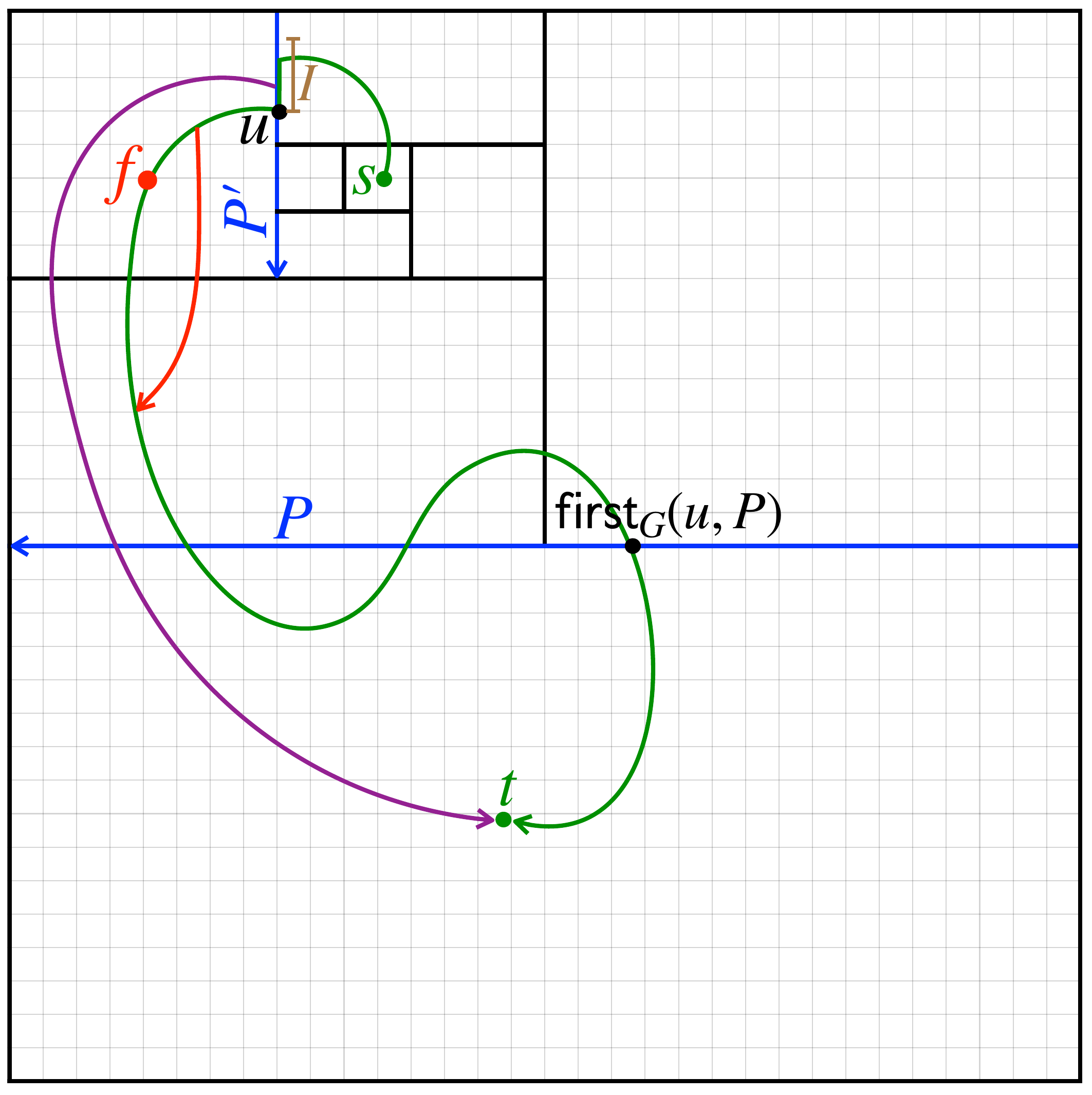}
  \caption{When $f \notin \Q$: The (green) $s$-to-$t$ path in $G$ first touches the (blue) separator path $\Psf\in \partial\Hsf$, then continues along $\Psf$ to $\u$, then goes through $f$ to $\first{G}{u}{\P}$ on the (blue) separator path $\P \in \partial H$, and from there to $t$ (without crossing $\P$ again). The (brown) interval $I$ contains all vertices $p$ that can reach $\first{G}{p}{\P}$ through $u$. When $f$ fails, the replacement $s$-to-$t$ path in $\Gf$ either takes a (red) detour around $f$ (and then $f$ stores this information) or is (purple) completely disjoint from the original (green) $u$-to-$\first{G}{u}{\P}$ part (and then $s$ stores this information).\label{fig:fnotinQ}}
   \end{center}
\end{figure}

Notice that if $\first{G}{s}{\P} \neq \first{G}{f}{\P}$ then the path in $G$ from $s$ to $\first{G}{s}{\P}$ does not go through $f$ and so $\first{\Gf}{s}{\P} = \first{G}{s}{\P}$. To check this, \hlgray{every vertex $s$ in $G$, for each of the $O(\log n)$ ancestor paths $\P$ of $s$, stores $\first{G}{s}{\P}$.} 

We handle the other case, where  $\first{G}{s}{\P} = \first{G}{f}{\P}$ as follows.
Consider a path in $G$ from $s$ to $\first{G}{s}{\P}$.  
We can choose such a path that: (1) begins with a prefix that is contained in $\Hsf$ from $s$ to $\first{\Hsf}{s}{\Psf}$ where $\Psf$ is some path of $\partial\Hsf$, then (2) continues along $\Psf$ until the last vertex $\u$ of $\Psf$ s.t $\first{G}{\u}{\P}  = \first{G}{f}{\P}$, and (3) ends with a suffix $S$ in $G$ from $\u$ to $\first{G}{f}{\P}$. See Figure~\ref{fig:fnotinQ}. Notice that the set of vertices $p\in \Psf$ s.t $\first{G}{p}{\P} = \first{G}{f}{\P}$ is a contiguous interval $\I$ of $\Psf$ ending in $\u$ (this is because any vertex of $\Psf$ earlier than $\u$ can reach $\first{G}{f}{\P}$ through $\u$). 
We think of all $p \in \I$ as reaching $\first{G}{\u}{\P}$ in $G$ using the same suffix $S$. 

\hlgray{For every vertex $f$ of $G$, for every pair of ancestor paths $\Psf,\P$ of $f$, for the maximal interval $\I$ of vertices $p \in \Psf$ such that $\first{G}{p}{P} = \first{G}{f}{P}$, if the path $S$ from the last vertex $u$ of $\I$ to $\first{G}{f}{P}$\footnote{If there are multiple paths in $G$ from $u$ to $\first{G}{f}{P}$ we fix and use an arbitrary such path.}
 goes through $f$, we let $f$ store the indices of the endpoints of $\I$. 
Every vertex $s$ of $G$ stores $\first{A}{s}{\Psf}$ for every ancestor piece $A$ of $s$ and every path $\Psf$ of $\partial A$.}     
This way, we can check if $\first{\Hsf}{s}{\Psf}$ is in the interval $\I$ (the interval stored by $f$ for the pair of paths $\Psf,\P$). If $\first{\Hsf}{s}{\Psf}$ is not in the interval $\I$, then $\first{\Gf}{s}{\P}=\first{G}{s}{\P}$ and we already have $\first{G}{s}{\P}$ stored in $s$. Otherwise, $\first{\Hsf}{s}{\Psf}$ is in the interval $\I$. Let $S_1$ be the prefix of $S$ ending just before $f$, and let $S_2$ be the suffix of $S$ starting immediately after $f$. Then, there are two options regarding the path in $\Gf$ from $s$ to $\first{\Gf}{s}{\P}$: 

\begin{enumerate}

\item The path intersects $S_2$. In this case, the path can continue along $S_2$ until reaching $\first{G}{f}{\P}$, so $\first{\Gf}{s}{\P}= \first{G}{f}{\P}$ and we have it stored in $f$. It only remains to check if this is indeed the case. Consider all vertices $p\in \I$ that can reach $\first{G}{f}{\P}$ in $\Gf$. They must constitute a (possibly empty) prefix of $\I$ (since any such vertex $p\in \I$ can reach any later vertex of $\I$ by going along $\I$). We therefore let \hlgray{$f$ store the index of the last vertex of the prefix of $\I$ that can reach $\first{G}{f}{\P}$ in $\Gf$}. This way, we only need to check if $\first{\Hsf}{s}{\Psf}$ (that is already stored in $s$) is earlier than this vertex.

\item The path is disjoint from $S_2$. To handle this case, we will identify two valid candidates for $\first{\Gf}{s}{\P}$ and take the earlier in $\P$ of these two candidates:

\begin{enumerate}
 
\item The path is disjoint from $S$. To handle this case, \hlgray{for every vertex $s$ of $G$, for every pair of ancestor paths $\Psf$,$\P$ of $s$, for the path $S$ defined by $\Psf,\P$ and $s$\footnote{To be clear, $S$ is the path from the last vertex $u$ of $\Psf$ s.t. $\first{G}{u}{\P}=\first{G}{s}{\P}$ to $\first{G}{s}{\P}$}, $s$ stores $\first{G\setminus S}{s}{\P}$ (i.e., the first vertex of $\P$ that is reachable from $s$ in $G$ using a path that is disjoint from $S$).}   
Since $f \in S$ then $f \notin  G\setminus S$ and so $\first{G\setminus S}{s}{\P}$ is a valid candidate for $\first{\Gf}{s}{\P}$. 

\item  The path intersects $S_1$, so it might as well go through $\u$ (the last vertex of $\I$). To handle this case, \hlgray{$f$ stores $\first{\Gf}{\u}{\P}$}, which is also a valid candidate for $\first{\Gf}{s}{\P}$.  

\end{enumerate}
\end{enumerate}

We summarize the case of $f \notin \Q$ with a general corollary that follows from the above.
\begin{corollary}\label{cor:4.1}
There is a (polylogarithmic-size) labeling to the vertices of $G$ that returns $\first{\Gf}{s}{\P}$ given the labels of any two vertices $s,f$ and the identity of any path $\P$ that (i) belongs to the boundary of an ancestor of $\Hsf$, and (ii) does not contain $f$. 
Symmetrically, we can find $\last{\Gf}{t}{\P}$ from the labels of any two vertices $t,f$.
\end{corollary}

\begin{figure}[htb]
  \begin{center}
 \includegraphics[scale=0.235]{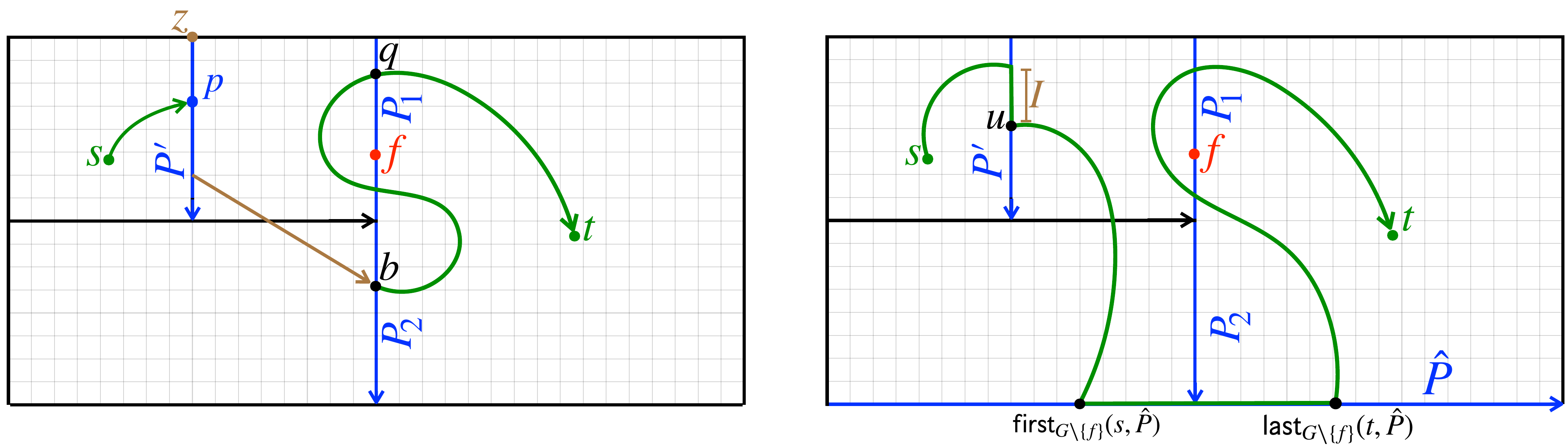}
   \caption{When $f\in \Q$: In the left example, the (green) replacement path $R$ touches only $\P\in \Q$ (the path on which $f$ lies). 
   In this case, $b$ is the first vertex of $\Pb$ that is reachable from $s$ in a path internally disjoint from $\Q$. The path from $b$ to $q=\first{\Gf}{b}{\Pa}$ is then found using the mechanism of Section~\ref{sec:singlePath}. 
   In the right example, the (green) replacement path $R$ touches some other path $\hat P \neq P$ (on which $f$ does not lie). In this case, using  Corollary~\ref{cor:4.1} we check if   
   $\first{\Gf}{s}{\PP} <_{\PP} \last{\Gf}{t}{\PP}$.
    \label{fig:touchesornot}}
   \end{center}
\end{figure}

\subsection{When $f\in \Q$} \label{sec:finQ}

Let $\P$ be the path of $\Q$ that contains $f$. 
Consider first the case where the replacement path $R$ touches some path $\PP\neq P$ of $\Q$ or of the boundary of some ancestor of $H$. Since boundary paths are vertex disjoint, $f\in \P$ implies $f \notin \PP$. Hence, in this case we can use Corollary~\ref{cor:4.1} twice, once to obtain $\first{\Gf}{s}{\PP}$, and once (using the symmetric part of~\cref{cor:4.1}) to obtain $\last{\Gf}{t}{\PP}$. Then, $s$ can reach $t$ in $\Gf$ iff $\first{\Gf}{s}{\PP} <_{\PP} \last{\Gf}{t}{\PP}$. See Figure~\ref{fig:touchesornot} (right). Therefore, in the remainder of this paper we deal with the case where other than $P$, $R$ does not touch any path of $Q$ or any path of the boundary of an ancestor of $H$. In this case, we cannot apply Corollary~\ref{cor:4.1} since $\P$ contains $f$.

Let $\Pa$ be the prefix of $\P$ ending just before $f$, and let $\Pb$ be the suffix of $\P$ starting immediately after $f$. 
Recall that $s$ is assumed to be in $\Htf$, the parent piece of $H$.
Consider the replacement path $R$. Since $R$ does not touch any path on the boundary of an ancestor of $H$, $R$ is contained in $\Htf^\circ = \Htf \setminus \partial \Htf$. 
$R$ starts with an $s$-to-$b$ prefix that is internally disjoint from $\Q$ and $b = \first{\Htf^\circ\setminus \Q}{s}{\Pa}$ or $b = \first{\Htf^\circ \setminus \Q}{s}{\Pb}$. Our first goal is to find these vertices $b$. Finding $b_1 = \first{\Htf^\circ\setminus \Q}{s}{\Pa}$ is easy since $b_1=\first{\Htf^\circ\setminus \Q}{s}{\P}$ (or else $b_1$ does not exist) so we let \hlgray{every vertex $s$ in $G$ store $\first{\Htf^\circ\setminus \Q}{s}{\P}$ for each of the $O(1)$ paths $P$ of each separator $Q$ of each of the $O(\log n)$ ancestor pieces $\Htf$ of $s$.} It therefore remains to find $b_2 = \first{\Htf^\circ\setminus \Q}{s}{\Pb}$. We note that if $s\in \Pa$ then we will not need $b_2$, and if $s\in \Pb$ then $b_2=s$.  

If either $s$ or $f$ is an ancestor apex of the other, then we store $b_2$ explicitly in either $s$ or $f$. That is, \hlgray{for every vertex $s$ (resp. $f$) of $G$, for every ancestor apex $f$ of $s$ (resp. $s$ of $f$), for every ancestor path $\P$ of $s$ (resp. $f$), if $f$ (resp. $s$) lies on $\P$ then $s$ (reps. $f$) stores in its label $b_2 = \first{\Gf}{s}{\Pb}$, where $\Pb$ is the suffix of $\P$ starting after $f$ (resp. $s$).} 

\begin{figure}[htb]
  \begin{center}
 \includegraphics[scale=0.235]{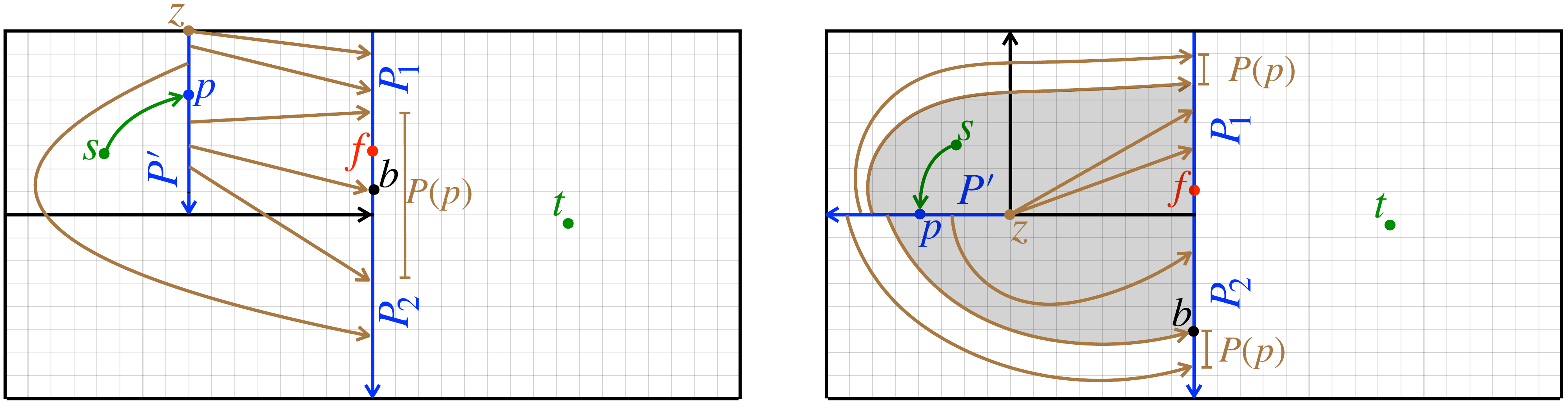}
   \caption{When $f\in \Q$: The vertical (blue) separator path $\P$ (in this example the separator $\Q$ consists of just a single path $\P$) is partitioned by $f$ into $\Pa$ and $\Pb$. The vertex $z$ is the first vertex of $p$'s (blue) path $\Psf$. The brown paths are the ways that $z$ can reach $\P$ in $\Htf^\circ\setminus \Q$. The vertex $b$ is the first vertex of $\Pb$ that is reachable from $p$ in $\Htf^\circ\setminus\Q$. In the left image, $f$ lies inside the interval of vertices $\Pz(p)$ that are reachable from $p$ ($b$ is therefore stored in $f$). In the right image, $f$ lies outside the two intervals $\Pz(p)$ (and $b$ is therefore stored in $s$). The shaded area is the cycle $C$ in the proof of Claim~\ref{claim:z}.  \label{fig:z}}
   \end{center}
\end{figure}

Recall the definition of $\Hsf'$ and $\Hsf$ from Section~\ref{sec:fnotinQ}. As in Section~\ref{sec:fnotinQ}, since we handled the case that $s$ or $f$ are apices, we can assume that $\Hsf'$ and $\Hsf$ are uniquely defined. 
Consider the $s$-to-$b_2$ path in $\Htf^\circ\setminus \Q$. It consists of a prefix that is internally disjoint from $\partial\Hsf$ and ends at a vertex of one of the $O(1)$ paths $\Psf\in \partial\Hsf$ (by definition of $\Hsf$, $f \notin \Hsf$, so $f$ does not belong to this prefix nor to $\Psf$). 
Let $p$ be the first vertex of $\Psf$ that is reachable in $\Htf^\circ$ from $s$ with a path internally disjoint from $\partial \Hsf$. Let $z$ be the first vertex of $\Psf$.   
Let $\Pz(z)$ denote the set of vertices of $\P$ that are reachable from $z$ in $\Htf^\circ$ with paths that are internally disjoint from $\Q$. Note that such paths in $\Htf^\circ \setminus \Q$ from $z$ to the vertices of $\Pz(z)$ all enter $\P$ from the same side (since they cannot touch $\partial\Htf$ nor the separator $\Q$). Moreover, 
since anything reachable in $\Htf^\circ \setminus \Q$ from $p$ is also reachable from $z$, we have that $\Pz(p)\subseteq \Pz(z)$. In fact, planarity dictates the following:
\begin{claim}\label{claim:z}
	The sequence $\Pz(p)$ consists of at most two intervals of consecutive vertices of $\Pz(z)$.
\end{claim}
\begin{proof}
Consider two vertices $u,v \in \Pz(p)$ that belong to two disjoint intervals of $\Pz(z)$. Let $C$ be the (undirected) cycle formed by the $p$-to-$v$ path, the $p$-to-$u$ path, and $\P[u,v]$ (see Figure \ref{fig:z}).
The vertex $z$ must be enclosed by $C$, otherwise, the path from $z$ to any vertex of $\P[u,v]$ must intersect the $p$-to-$v$ path or the $p$-to-$u$ path (and is thus reachable from $p$ as well in contradiction to the two intervals being disjoint). Since $z$ is enclosed by $C$, any vertex of $\P$ that is reachable from $z$ and does not belong to $\P[u,v]$ is also reachable from $p$ (since the path to it from $z$ must intersect the $p$-to-$v$ path or the $p$-to-$u$ path).
\end{proof}

\hlgray{For every $s \in G$, for every pair of ancestor pieces $\Hsf, \Htf$ of $s$, for every two paths $\Psf \in \partial \Hsf$, and $\P$ of the separator $\Q$ of $\Htf$, $s$ stores the identities of the endpoints of the two intervals of $\Pz(p)$ in $\Pz(s)$ (where $p$ and $z$ are as defined above).
For every $f \in G$, for every pair of ancestor pieces $\Hsf, \Htf$ of $f$, for every two paths $\Psf$ of the boundary of the sibling of $\Hsf$, and $\P$ of the separator $\Q$ of $\Htf$, $f$ stores the identity of the first vertex $v$ of $\Pz(z)$ that is after $f$ on $\Psf$ (where $z$ is the first vertex of $\Psf$).
}

We can finally describe how to find $b_2 = \first{\Htf^\circ\setminus \Q}{s}{\Pb}$: If the vertex $v$ stored in the label of $f$ falls inside one of the two intervals stored in the label of $s$ (for the common pair $\Psf,\P$) then $b_2=v$ and we have it stored in $f$. Otherwise, $b_2$ is the earliest starting endpoint that is later than $f$ among the two endpoints  of the two intervals stored in the label of $s$ for $\Psf$ and $\P$ (if both intervals are before $f$ on $\P$ then $b_2$ does not exist).

We have therefore achieved our first goal: we found $b_1 = \first{\Htf^\circ\setminus \Q}{s}{\Pa}$ and $b_2 = \first{\Htf^\circ\setminus \Q}{s}{\Pb}$. 
It remains to check if there is a replacement path $R[b,t]$ in $\Gf$ from some $b\in \{b_1,b_2\}$. We can assume that such a path starts from $b$, continues to $q=\first{\Gf}{b}{\Pa}$ or to $q=\first{\Gf}{b}{\Pb}$, continues from $q$ to $\last{\Gf}{t}{\Pa}$ (if $q\in \Pa$) or to $\last{\Gf}{t}{\Pb}$ (if $q\in \Pb$), and then finally continues to $t$. See Figure~\ref{fig:touchesornot} (left).
In Section~\ref{sec:singlePath} we present a labeling scheme that labels the vertices of $\P$ s.t. given the labels of any $b,f\in \P$ we can find both $\first{\Gf}{b}{\Pa}$ and $\first{\Gf}{b}{\Pb}$. A symmetric labeling finds $\last{\Gf}{t}{\Pa}$ and  $\last{\Gf}{t}{\Pb}$. 
We shall refer to this labeling as the {\em secondary} labeling. \hlgray{The secondary labels are stored in the labels of $s$ and $f$ along with the identities of the endpoints of the intervals (for $s$) and the first vertex $v$ (for $f$) as described above.}
This ways, the secondary label of $b_1$ is available in the label of $s$ and the secondary label of $b_2$ is available in the label of $f$ (if $b_2=v$ falls inside one of the two intervals stored in $s$) or of $s$ (if $b_2$ does not fall inside one of the two intervals). Recall that if $s$ or $f$ is an apex then we stored $b_2$ explicitly (in either $s$ or $f$), in this case we additionaly store the secondary label of $b_2$. 

Before moving on to describe the secondary labeling scheme, there is one delicate issue: This labeling scheme requires that the two endpoints of the path $\P$ lie on the same face. In our case, if the separator $Q$ consists of just the path $P$ (as in the examples in the figures), then the endpoints of $P$ already lie on a single face of $\Htf$ (the boundary of $\Htf$). Otherwise, recall that we are working under the assumption that the replacement path $R$ does not touch any path of the cycle separator $\Q$ other than $\P$. 
Hence, before computing the secondary labels we make an incision along the edges of $\Q$ that do not belong to $\P$. Under our assumption the incision does not affect the replacement path, and now the endpoints of $\P$ indeed lie on a single face.  

\section{When the Query Vertices Lie on a Known Path}\label{sec:singlePath}
In this section we present the secondary labeling scheme  (with polylogarithmic-size labels), which addresses the following problem: We are given a directed planar graph $G$ and a single path $\P$ in $G$ whose endpoints lie on the same face. We need to label the vertices of $\P$ such that given the labels of any two vertices $b,f$ of $\P$ we can find the first vertex $p < f$ of $\P$ that is reachable from $b$ in $\Gf$ and the first vertex $p > f$ of $\P$ that is reachable from $b$ in $\Gf$.

\subsection{An auxiliary procedure}\label{sec:auxiliary}
We begin with an auxiliary procedure that will be useful for our labeling. 
In this procedure, we wish to label the vertices of $\P$, such that given the labels of any two vertices $b,f$ such that $f$ is before $b$ on $P$ (i.e., $f < b$), we can find the first vertex $p>f$ of $P$ that is reachable in $G$ from $b$ using a path that does not touch $f$ or any vertex before $f$ (i.e., does not touch any vertex $v \leq f$ of $\P$).

Let $H_P$ be the graph composed of the path $\P$ and the following additional edges: for every pair of vertices $u,v \in \P$ where $u>v$, we add an edge $(u,v)$ iff (1) there exists a $u$-to-$v$ path in $G$ that does not touch $\P$ before $v$, and (2) there is no such $w$-to-$v$ path for any $w>u$. The following claim shows that instead of working with $G$, we can work with $H_\P$:

\begin{claim}\label{claim:HP}
	Given any $f < b$, the first vertex $p>f$ that is reachable from $b$ using a path that does not touch $f$ or any vertex of $\P$ before $f$, is the same in $G$ and in $H_P$. 
\end{claim}
\begin{proof}
Let $S$ be the corresponding $b$-to-$p$ path in $G$ for $f<p<b$.
The path $S$ visits no vertex of $\P$ before $f$ (by definition of $S$) and no vertex of $\P$ between $f$ and $p$ (by definition of $p$). Hence, in $H_P$ there is an edge $(u,p)$ for some $u>b$, and so $S$ is represented in $H_P$ (by a path composed of a $b$-to-$u$ prefix along $\P$ followed by the single edge $(u,p)$). In the other direction, let $R$ be the corresponding $b$-to-$p$ path in $H_P$. The path $R$ visits no vertex of $\P$ before $f$ (by definition of $R$) and every edge $(u,v)$ of $R$ corresponds to a path in $G$ that visits no vertex of $\P$ before $v$ (and therefore visits no vertex of $\P$ before $f$), hence $R$ is appropriately represented in $G$. 
\end{proof}

We call the edges of $H_P$ that are not edges of $\P$ {\em detours}. 
We will use the fact that detours do not cross (i.e., they form a laminar family):

\begin{claim}
If there is a detour $(u,v)$ then there is no detour  $(w,x)$ with $v<x<u<w$.  
\end{claim}
\begin{proof}
By concatenating the $(w,x)$ detour, the $x$-to-$u$ subpath of $\P$, and the $(u,v)$ detour, we get a path in $G$ that does not touch $\P$ before $v$ but starts at a vertex $w>u$ contradicting condition (2) of the $H_P$ edges.  
\end{proof}

We now explain what needs to be stored to facilitate the auxiliary procedure. 
We define the {\em size} $|d|$ of a detour $d=(u,v)$ as the number of vertices of $\P$ between $v$ and $u$. 
We say that a vertex of $\P$ is {\em contained} in a detour $d=(u,v)$ if it lies on $\P$ between $v$ and $u$. We say that a detour $d'$ is {\em contained} in detour $d$ if all vertices that are contained in $d'$ are also contained in $d$. Notice that, from Claim~\ref{claim:HP}, the vertex $p$ sought by the auxiliary procedure is an endpoint of the largest detour $d$ that contains $b$ and does not contain $f$. To find $d$, \hlgray{every vertex $v$ of $\P$ stores a set of $O(\log n)$ nested detours $d^v_1,d^v_2,\dots$ where $d^v_1$ is the largest detour containing $v$, and $d^v_{i+1}$ is the largest detour of size $|d^v_{i+1}| \le |d^v_i|/2$ containing $v$. Additionally, for each $d^v_i$, $v$ stores the largest detour $\hat{d}{^v_i}$ that is strictly contained in $d^v_i$ but does not contain $d^v_{i+1}$.} See Figure \ref{fig:nested}.

\begin{figure}[htb]
  \begin{center}
 \includegraphics[scale=0.3]{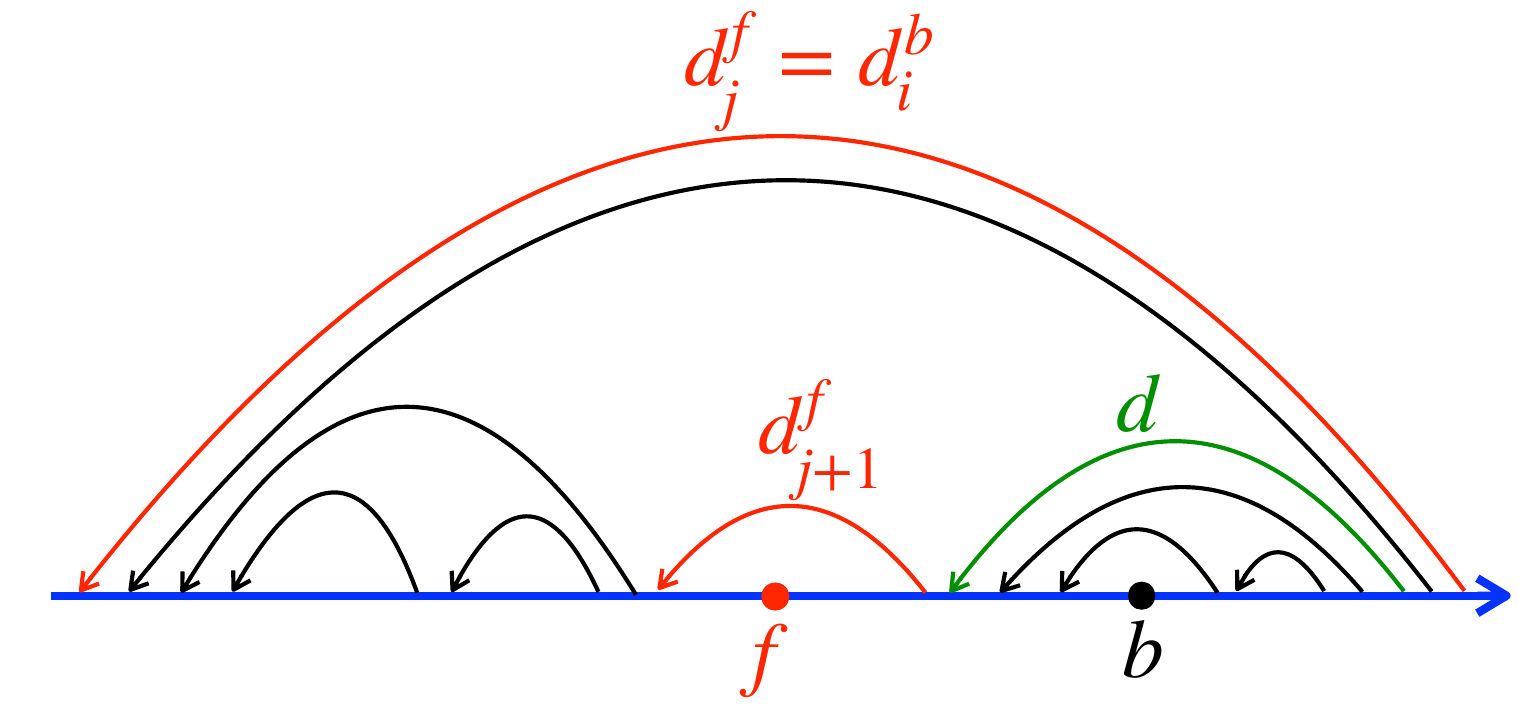}
  \caption{A nested system of $O(\log n)$ detours  $d^f_1,d^f_2,\dots$ of $f$ (in red) used in the auxiliary procedure. The (green) detour $d$ is the largest detour  that contains $b$ and does not contain $f$. Either $d=\hat{d}{^f_j}$  and is stored in $f$, or $d=d^b_{i+1}$ and is stored in $b$.
  \label{fig:nested}}
   \end{center}
\end{figure}

Consider the smallest detour $d^f_j=d^b_i$ that is saved by both $b$ and $f$. If such a detour does not exist then there is no detour that contains both $b$ and $f$ and so  $d= d^b_1$ is stored in $b$. If $|d|>|d^f_j|/2$ then $d$ must be the largest detour that is contained in $d^f_j$ but does not contain $d^f_{j+1}$. In other words, $d=\hat{d}{^f_j}$ and is stored in $f$. If on the other hand $|d|\le |d^f_j|/2$, then $|d|\le |d^b_i|/2$ and by the choice of $d^b_i$ we have that $d^b_{i+1}$ does not contain $f$ and so $d=d^b_{i+1}$ and is stored in $b$. This completes the description of the auxiliary procedure. 

\subsection{The labeling}
Equipped with the above auxiliary procedure, we are now ready to describe the secondary labeling scheme. Recall that there are four cases to consider: Given $b,f \in \P$ where $b$ can be before/after $f$ we wish to find the first vertex $p$ before/after $f$ that is reachable from $b$ in $\Gf$. There are four cases to consider: 

\paragraph{Given $f < b$, find the first vertex $p<f$ that is reachable from $b$ in $\Gf$.} Consider the $b$-to-$p$ path $R$ in $\Gf$. Let $r$ be the first vertex of $R$ that belongs to $\P$ and $r<f$. Let $r'$ be the vertex of $\P$ that precedes $r$ on $R$. Note that $r<f$, that $r'>f$, and that the $r'$-to-$r$ subpath of $R$ (which we call a {\em bypass} of $f$ and denote $R[r',r]$) is internally disjoint from $\P$ and it either emanates to the left or to the right of $\P$. Moreover, since the endpoints of $\P$ lie on the same face, then if $R$ emanates at $r'$ to the left (resp. right) of $\P$ it must enter $\P$ at $r$ from the left (resp. right) of $\P$. This means that we can assume w.l.o.g that the bypass $R[r',r]$ is the largest such bypass (in terms of the number of vertices of $\P$ between $r$ and $r'$). This is because any smaller bypass that is on the same side of $\P$ as $R[r',r]$ is either contained in $R[r',r]$ (in which case we might as well use $R[r',r]$) or intersects $R[r',r]$ implying (in contradiction) that there is a larger bypass.  See Figure~\ref{fig:bypass}.
 
 \begin{figure}[htb]
  \begin{center}
 \includegraphics[scale=0.3]{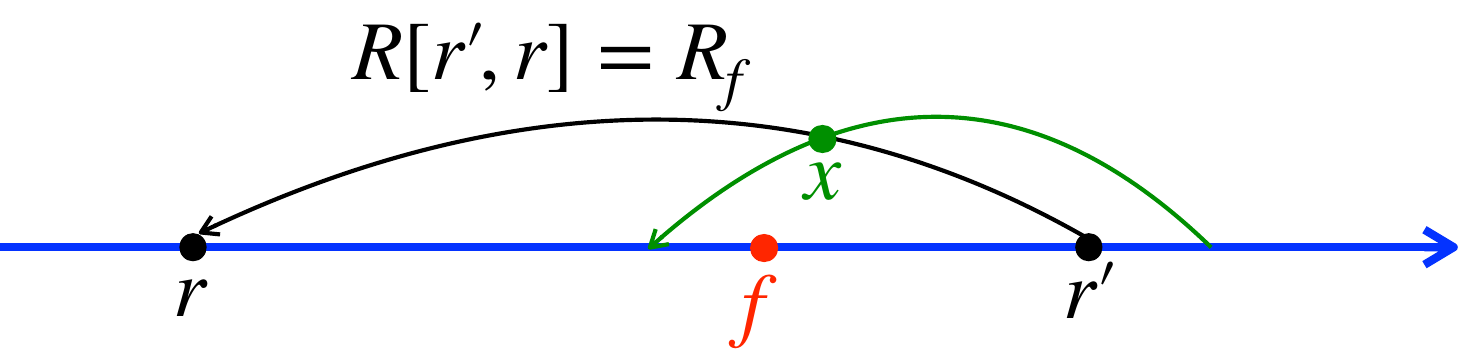}
  \caption{A bypass $R[r',r]$ (in black) of $f$. 
  If there was an intersecting bypass (in green) that is larger than $R[r',r]$ then the green subpath before $x$ and the black subpath after $x$ would constitute a larger bypass.  \label{fig:bypass}}
   \end{center}
\end{figure}

 We let \hlgray{each vertex $f \in \P$ store the endpoints of the largest bypass $L_f$ (resp. $R_f$) that is internally disjoint from $\P$, emanates left (resp. right) of $\P$ at a vertex that is after $f$, and enters $\P$ from the left (resp. right) at a vertex that is before $f$. The vertex $f \in \P$ also stores the first vertex $L_f^+$ (resp. $R_f^+$) of $\P$ that is before $f$ and is reachable in $\Gf$ from the endpoint of $L_f$ (resp. $R_f$) that is after $f$.}
 
In order to find $p$, we first use the auxiliary procedure of Section~\ref{sec:auxiliary} to find the first vertex $p'>f$ that is reachable in $G$ from $b$ using a path that does not touch any vertex of $\P$ before $f$. Then, we check whether $p'$ is contained in $L_f$ and if so we consider $L_f^+$ as a candidate for $p$. Similarly, we check whether $p'$ is contained in $R_f$ and if so we consider $R_f^+$ as a candidate for $p$. Finally, we return the earlier of the (at most two) candidates.

 \paragraph{Given $f < b$, find the first vertex $p>f$ that is reachable from $b$ in $\Gf$.}
This case is very similar to the previous case. The only differences are: (1) we let \hlgray{every vertex $f \in \P$ store the first vertex $L_f^-$ (resp. $R_f^-$) of $\P$ that is after $f$ and is reachable in $\Gf$ from the endpoints of $L_f$ (resp. $R_f$)}, and (2) we add $p'$ itself as a third possible candidate for $p$.

 \paragraph{Given $b < f$, find the first vertex $p>f$ that is reachable from $b$ in $\Gf$.}
 This case and the next one are handled by small (but not symmetric) modifications of the previous two cases.
 Consider the $b$-to-$p$ path $R$ in $\Gf$. Let $r$ be the first vertex of $R$ that belongs to $\P$ and $r>f$. Let $r'$ be the vertex of $\P$ that precedes $r$ on $R$. The subpath $R[r',r]$ (which we call a {\em byway} of $f$) is internally disjoint from $\P$, and if it emanates at $r'$ to the left (resp. right) of $\P$ then it must enter $r$ from the left (resp. right) of $\P$. We can assume w.l.o.g that  $R[r',r]$ is the smallest such byway (because any larger byway that is on the same side of $\P$ as $R[r',r]$ either contains $R[r',r]$ (in which case we might as well use $R[r',r]$) or intersects $R[r',r]$ implying (in contradiction) that there is a smaller byway. See Figure~\ref{fig:byway}.
 
 We let \hlgray{each vertex $f \in \P$ store the endpoints of the smallest byway $L_f$ (resp. $R_f$) containing $f$ that is to the left (resp. right) of $\P$. The vertex $f \in \P$ also stores the first vertex $L_f^-$ (resp. $R_f^-$) of $\P$ that is after $f$ and is reachable in $\Gf$ from the endpoint of $L_f$ (resp. $R_f$) that is after $f$.} 
 
In order to find $p$, we begin by finding the first vertex $p'<f$ that is reachable in $G$ from $b$ using a path that does not touch any vertex of $\P$ after $f$. This is done by using a variant of the auxiliary procedure of Section~\ref{sec:auxiliary} with the following modification. In $H_P$ we add the detour $(u,v)$ iff (1) there exists a $u$-to-$v$ path in $G$ that does not touch $\P$ before $v$, and (2) there is no such $u$-to-$w$ path for any $w<v$.
After finding $p'$ using this mechanism, we check whether $p'$ appears on $P$ before the starting point of the byway $L_f$, and if so we consider $L_f^-$ as a candidate for $p$. Similarly, we check whether $p'$ appears on $P$ before the starting point of the byway $R_f$ and if so we consider $R_f^-$ as a candidate for $p$. Finally, we return the earlier of the (at most two) candidates.

 \begin{figure}[htb]
  \begin{center}
 \includegraphics[scale=0.3]{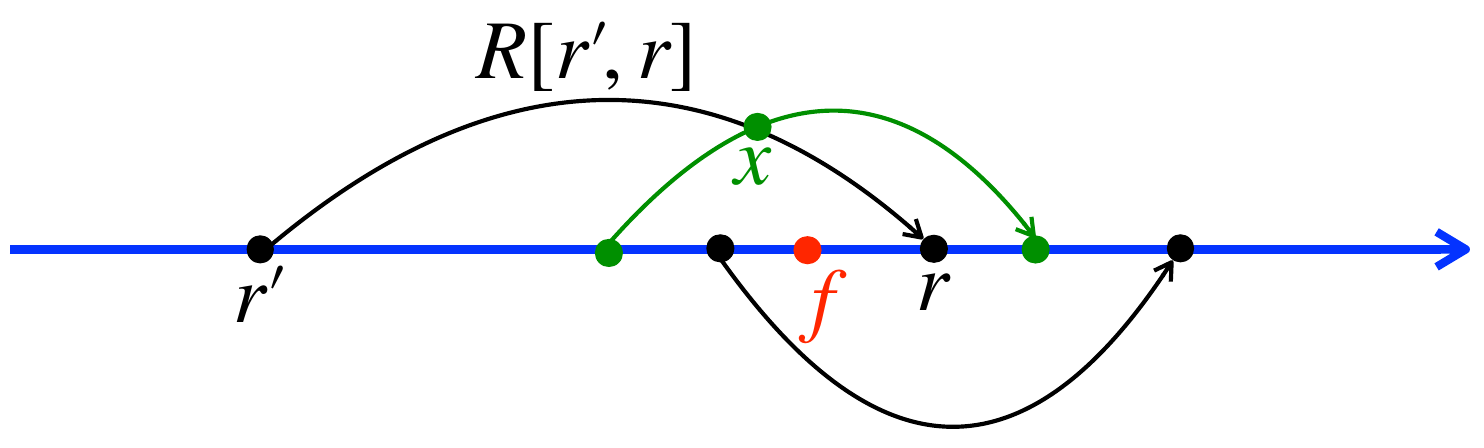}
  \caption{Two byways of $f$ (in black), the top one is $R[r',r]$. If there was an intersecting byway (in green) that is smaller than $R[r',r]$ then the green subpath before $x$ and the black subpath after $x$ would constitute a smaller byway. \label{fig:byway}}
   \end{center}
\end{figure}

 \paragraph{Given $b < f$, find the first vertex $p<f$ that is reachable from $b$ in $\Gf$.}
 This case is very similar to the previous case. The only differences are: (1) we \hlgray{let every vertex $f \in \P$ store the first vertex $L_f^+$ (resp. $R_f^+$) of $\P$ that is before $f$ and is reachable in $\Gf$ from the endpoints of $L_f$ (resp. $R_f$),} and (2) we add $p'$ itself as a third possible candidate for $p$.

\bibliographystyle{abbrv}

\end{document}